\documentclass[letterpaper,11pt]{article}
\usepackage[utf8]{inputenc}

\pdfoutput=1

\usepackage{amsmath, amsthm, amssymb}
\usepackage{algpseudocode,algorithm,algorithmicx}
\usepackage{mathtools}
\usepackage[numbers]{natbib}
\usepackage{comment} 
\usepackage{color-edits} 
\usepackage{tcolorbox}
\usepackage{xfrac}
\usepackage{hyperref}
\usepackage{authblk}
\usepackage{multirow}
\usepackage{caption}
\usepackage{bm}
\usepackage{newfloat}

\usepackage[T1]{fontenc}

\usepackage[margin=1in]{geometry}

\definecolor{mygreen}{RGB}{20,120,60}

\hypersetup{
     colorlinks=true,
     citecolor= mygreen,
     linkcolor= mygreen
}

\title{{\em Semi-MapReduce} Meets Congested Clique}

\author[1]{Soheil Behnezhad}
\author[1]{Mahsa Derakhshan}
\author[1]{MohammadTaghi Hajiaghayi}
\affil[1]{Department of Computer Science, University of Maryland\protect\\ \texttt{\{soheil, mahsa, hajiagha\}@cs.umd.edu}}
\date{}

\begin{document}


\newcommand{\Ot}[1]{\ensuremath{\widetilde{O}(#1)}}
\newcommand{\Tht}[1]{\ensuremath{\widetilde{\Theta}(#1)}}
\newcommand{\Omt}[1]{\ensuremath{\widetilde{\Omega}(#1)}}
\renewcommand{\O}[1]{\ensuremath{O(#1)}}
\newcommand{\Th}[1]{\ensuremath{\Theta(#1)}}
\newcommand{\Om}[1]{\ensuremath{\Omega(#1)}}

\newcommand{\CC}[0]{congested clique}
\newcommand{\CONGEST}[0]{\ensuremath{\mathsf{CONGEST}}}
\newcommand{\MPC}[0]{\ensuremath{\mathsf{MPC}}}
\newcommand{\semiMPC}[0]{\ensuremath{\mathsf{semi}\text{-}\mathsf{MPC}}}

\addauthor{sb}{mygreen}    
\addauthor{md}{red}  

\newtheorem{theorem}{Theorem}[section]
\newtheorem{lemma}[theorem]{Lemma}
\newtheorem{proposition}[theorem]{Proposition}
\newtheorem{corollary}[theorem]{Corollary}
\newtheorem{conj}[theorem]{Conjecture}
\newtheorem{definition}[theorem]{Definition}
\newtheorem{claim}[theorem]{Claim}
\newtheorem{fact}[theorem]{Fact}
\newtheorem{observation}[theorem]{Observation}
\newtheorem{remark}[theorem]{Remark}
\newtheorem{assumption}[theorem]{Assumption}

\newcommand{\etal}[0]{\textit{et al.}}

\newcommand{\presectionspace}[0]{}
\newcommand{\negspace}[0]{}

\maketitle
\begin{abstract}
Graph problems are troublesome when it comes to MapReduce. Typically, to be able to design algorithms that make use of the advantages of MapReduce, assumptions beyond what the model imposes, such as the {\em density} of the input graph, are required.

In a recent shift, a simple and robust model of MapReduce for graph problems, where the space per machine is set to be \O{|V|}, has attracted considerable attention. We term this model {\em semi-MapReduce}, or in short, \semiMPC{}, and focus on its computational power.

We show through a set of simulation methods that \semiMPC{} is, perhaps surprisingly,  equivalent to the congested clique model of distributed computing. However, \semiMPC{}, in addition to round complexity, incorporates another practically important dimension to optimize: the number of machines. Furthermore, we show that algorithms in other distributed computing models, such as \CONGEST{}, can be simulated to run in the same number of rounds of \semiMPC{} while also using an optimal number of machines. We later show the implications of these simulation methods by obtaining improved algorithms for these models using the recent algorithms that have been developed.
\end{abstract}

\section{Introduction}

The MapReduce framework is used persistently at large- and small-scale companies to analyze massive data sets. As such, several theoretical models have been proposed to analyze the computational power of such systems. {\em Massively parallel computations} (\MPC{}) is arguably the most popular model that captures the essence of MapReduce-like models while abstracting away their technical details (see Section~\ref{sec:model} for its formal definition).

We consider the most restrictive \MPC{} model that allows for fruitful graph algorithms. Let us denote by $n$ the number of vertices of the input graph and by $m$ the number of edges. This model, which we term {\em semi-MapReduce} or, in short \semiMPC{}, restricts the space of the machines to be \O{n}. This is similar, in spirit, to the well studied semi-streaming model (see \cite{feigenbaum2005graph}). This leaves two parameters to be determined by the algorithm: the round complexity and the number of machines. We are, generally, interested in constant or sub-logarithmic round algorithms that use \O{n} or optimally\footnote{Optimal in the sense that there is only enough total space to store the input graph.} \O{m/n} machines. To date, even for simple graph problems such as graph connectivity, no sublogarithmic round algorithm is known that uses sublinear space and a sublinear number of machines. This has even been the base for some impossibility results (see e.g., Conjecture 3.1 of \cite{yaroslavtsev2017massively}). Thus, the \O{n} space in \semiMPC{} is in some sense the best possible.

In this paper, we focus on the computational power of \semiMPC{} and compare it with other well-studied models. We give simulation methods that show \semiMPC{} is  equivalent to the congested clique model of distributed computing in terms of the number of rounds that it takes to solve a problem. These methods lead to important results for both of these models and uncover a perhaps surprising relevance. For instance, one can obtain an $O(1)$ round \semiMPC{} algorithm for minimum spanning tree (or other related problems such as graph connectivity) by simulating the very recent congested clique algorithm of Jurdzi\'{n}ski and Nowicki \cite{jurdzinski2018mst}. On the other hand, the recently developed algorithms that fit into \semiMPC{}, e.g., the algorithms of Czumaj \etal{}~\cite{czumaj2017round} and Assadi \etal{}~\cite{assadi2017coresets}, can be simulated on the congested clique model to obtain the first sub-logarithmic round algorithms for approximate maximum matching and approximate vertex cover problems. We further show how to simulate algorithms of the \CONGEST{} model of distributed computing in the same number of rounds of \semiMPC{} using an optimal number of machines.

\subsection{Related Work}
There have been many papers in the literature that consider the computational power of MapReduce and other distributed models \cite{roughgarden2016shuffles, drucker2014power, karloff2010model, hegeman2015lessons}. Most relevant to the present work is the paper of Hegeman et al.~\cite{hegeman2015lessons} that gives a simulation method for running congested clique algorithms on MapReduce for dense graphs (and not the other way around). Our focus, in this paper, is specifically on the \semiMPC{} model, and show it is equivalent (i.e., a two-way simulation method) to the congested clique model and other distributed computing models. 

We note that ``proximity'' of congested clique and \semiMPC{} has also been independently mentioned in the very recent paper of Ghaffari \etal{}~\cite{ghaffari2018improved}. Our simulation methods here, independent from the problem at hand, formalize the {\em equivalence} of these two models.

\section{The \semiMPC{} Model}\label{sec:model}
We start with the formal definition of the {\em massively parallel communications} model which was first introduced in \cite{beame2013communication}.


\negspace{}
\paragraph{The \MPC{} model.} An input (not necessarily a graph) of size $\ell$ is initially distributed among $p$ machines (processors) each with a space of size $s$. The total space of all machines (i.e., $p \cdot s$) is bounded by $\Ot{\ell ^ {1 + \delta}}$ for a small parameter $\delta \in [0, 1)$ of the algorithm that controls \emph{replication}. Furthermore, the number of machines should not be more than the space per machine\footnote{This is motivated by a practical restriction that each machine should be able to store the identifier (e.g., IP) of all other machines to be able to communicate with them.} (i.e., $p \leq s$). Computation proceeds in synchronous \emph{rounds}: In each round, each machine performs a local computation on its data, and at the end of the round, communicates with other machines. The total size of the messages that a machine sends or receives in each round should not be more than its space $s$.

Round complexity is practically the dominant cost of computation in MapReduce systems (see e.g., \cite{beame2013communication, lattanzi2011filtering}). Thus, the primary goal for \MPC{} algorithms is to minimize the number of rounds while keeping the number of machines and their local space \emph{significantly} smaller than the input size. That is, in most cases, we desire algorithms that use $\O{\ell^{1-\epsilon_1}}$ space and $\O{\ell^{1-\epsilon_2}}$ machines for constants $\epsilon_1, \epsilon_2 > 0$.

We are now ready to formalize \semiMPC{} for graph problems.

\negspace{}
\paragraph{The \semiMPC{} model.} Let $G = (V, E)$ with $n = |V|$ and $m = |E|$, be the input graph. An algorithm is in \semiMPC{} if it is in \MPC{} and each machine uses a space of size $s = \O{n}$.

Note that by $\O{n}$ space, we mean $\O{n}$ words of size $\O{\log n}$ bits. That is, the machines are able to store the index of all the vertices. Since \semiMPC{} fixes $s$, only two parameters are left to be determined by the algorithm: The round complexity and the number of machines. Generally we desire algorithms that run in constant or sublogarithmic rounds.


\negspace{}
\paragraph{Advantages of \semiMPC{}.} A main determining factor in the round complexity of graph algorithms in \MPC{} is whether it is possible to store the vertices of the input graph in one machine. By fixing the space to be \O{n}, the \semiMPC{} model allows storage of all vertices in one machine but prevents the algorithm to use a significantly more space. Compared to traditional \MPC{} algorithms that impose the space restriction based on the input size (i.e., the number of edges of the input graph), \semiMPC{} is more restrictive for {\em dense} graphs and less restrictive for {\em sparse} graphs. More precisely, for  dense graphs, where $m = \Omega(n^{1+c})$ for some constant $c > 0$, \semiMPC{} imposes a strict upper bound of \O{n} on the space per machine whereas anything sublinear in $m$ (e.g., $\O{n^{1+c/2}}$) would be acceptable for an \MPC{} algorithm. However, for sparse graphs, where $m$ is almost linear in $n$, the space per machine in \semiMPC{} is no longer sublinear in the input size and it becomes easier to design \semiMPC{} algorithms rather than \MPC{} algorithms.

As it was already mentioned, the insufficiency of space to store all the vertices in each machine for \MPC{} is one of the main reasons that make solving graph problems particularly hard on it. This is captured by the widely accepted conjecture that graph connectivity (or even the simpler problem of deciding if the input graph is one cycle or two cycles) takes $\Om{\log n}$ rounds of \MPC{} if $s = n^{1-\epsilon}$ for any constant $\epsilon > 0$. The matching upper bound of \O{\log n} for many graph problems could generally be achieved by simply simulating PRAM algorithms \cite{karloff2010model, goodrich2011sorting}. This indicates that the local free computation of \MPC{} does not help in reducing the round complexity of graph algorithms. On the contrary, many graph algorithms could be solved in sublogarithmic rounds of \semiMPC{}. For example, our simulation methods in Section~\ref{sec:simul} prove that graph connectivity and minimum spanning tree problems can be solved in $O(1)$ rounds of \semiMPC{}.

\negspace{}
\section{Comparison \& Simulations}\label{sec:simul}
We start with the {\em congested clique} \cite{peleg2000distributed} model. In the congested clique model of distributed computing, each node of the graph is initially aware of its incident edges. Then in synchronous communication rounds, every pair of nodes, even if not connected in the input graph, can exchange a message of size $\O{\log n}$ bits.

Since congested clique limits the message size between every pair of the nodes to be at most \O{\log n} bits, it becomes quite straight-forward to simulate it in the seemingly more powerful \semiMPC{} model which allows the machines to communicate up to \O{n} messages of length $\O{\log n}$ bits. This is also true for \MPC{} as long as the machines have access to a memory of size \Om{n}, i.e., for dense graphs (as it was mentioned in \cite{hegeman2015lessons}).

\begin{theorem} \label{thm:cc-mpc}
	Let $\mathcal{A}_{CC}$ be an algorithm that for an input graph $G=(V, E)$, with $n$ vertices, runs in $T$ rounds of \CC{} using a local memory of $\O{n}$ in each node. One can simulate $\mathcal{A}_{CC}$ in $\O{T}$ rounds of \semiMPC{} using $n$ machines.
\end{theorem}
\begin{proof}
	The first difference between congested clique and the \semiMPC{} model is in the way that the input is initially distributed to the nodes/machines. In the congested clique model, each node initially knows its edges. On the other hand, in the \semiMPC{} model, the edges are initially distributed adversarially into the machines.
	
	We assume that we have access to $n$ machines of \semiMPC{}. Therefore we simulate each of the nodes of congested clique on one machines. That is, in the first round, each machine, for any edge $(u, v)$ that is stored in it, sends a message to the corresponding machine of $u$ and a message to the corresponding machine of $v$ informing them about the edge. It is easy to confirm that the number of messages that each machine sends and receives is $\O{n}$ since each node has degree at most $n$ and that the number of edges that are initially stored in the machines is of $\O{n}$.
	
	Next, in each round, the \semiMPC{} machines exactly simulate the internal computation of their corresponding node in the congested clique algorithm and send the same message that their corresponding node would send. There are two things that we would have to argue would not be violated: (1) The number of messages that the machines send and receive; (2) The internal space. The first would not be more than $\O{n}$ since each node of the congested clique can send and receive at most $\O{n}$ messages. The latter also would not be violated because we explicitly assumed in the theorem statement that the congested clique nodes use an internal space of size $\O{n}$.
	
	Overall it takes only $T + 1$ rounds to simulated any $T$ round congested clique algorithm that satisfies the mentioned properties on \semiMPC{}.
\end{proof}

Note that it is essential to explicitly mention that the congested clique algorithm should use a space of size $\O{n}$ on each machine. In theory, an algorithm on congested clique might require a significantly more space on one node than the total data that is sent to it from other nodes and this cannot be simulated on \semiMPC{}; this, however, rarely happens in natural algorithms. It is worth mentioning that it is also common to only restrict the communication size of the machines in the \MPC{} model instead of the amount of space. This is, for instance, how the original paper of Beame~\cite{beame2013communication} defines \MPC{}. With this definition, we can get rid of the extra condition that the congested clique algorithm has to use \O{n} local space.

While it is not possible to provide a better guarantee on the number of machines required in Theorem~\ref{thm:cc-mpc} in the general case, we remark that many natural algorithms on the congested clique require much fewer number of machines when simulated on the \semiMPC{} model. The reason is that typically, techniques such as coresets, sampling, linear sketches, etc., perform the main computation on very few -- and usually just a single -- node and use the other nodes to only store the data. In such scenarios, we require only $O(m/n)$ (i.e., optimal) number of \semiMPC{} machines to simulate the algorithm.

The more challenging direction is to simulate any \semiMPC{} algorithm in the congested clique model for which we use the routing scheme of Lenzen \cite{lenzen2013optimal}.

\begin{theorem}\label{thm:mpc-cc}
	Let $\mathcal{A}_M$ be an algorithm that for an input graph $G$, with $n$ vertices, runs in $T$ rounds of \semiMPC{}. One can simulate $\mathcal{A}_M$ in $\O{T}$ rounds of \CC{}.
\end{theorem}
\begin{proof}
	To prove this, we simulate each machine in \semiMPC{} by a \CC{} node. In each round of \semiMPC{} each machine sends and receives at most $\O{n}$ messages of size $\O{\log n}$ bits. Each message can be represented by its sender, receiver and its value. 
	
	The challenge is that in the congested clique, because of the bandwidth limit, we cannot directly send $\O{n}$ messages from one node to another. However, there is a simple workaround for this. Consider a simpler example where one node wants to send $\O{n}$ messages to another one. We can first distribute the message to all other nodes in one round (i.e., by sending $\O{n}$ messages of size $\O{\log n}$ bits to the other nodes), then each node in the next round, sends this message to the destination node. This means that without violating the bandwidth, one machines can send a message of size $\O{n}$ to another in 2 rounds. However, note that we require a more complicated routine for simulating \semiMPC{} on the congested clique. Here, {\em each} machine can send and receive up to $\O{n}$ messages at the same time. Fortunately, there is a routing scheme by Lenzen~\cite{lenzen2013optimal} (Theorem 3.7) that can be used for this. More precisely, \cite{lenzen2013optimal} shows that if each node is the destination and source of $\O{n}$ messages, there exists a deterministic $\O{1}$ round algorithm that delivers them without violating the bandwidth limit. Therefore, we can simulate any algorithm that takes $T$ rounds of \semiMPC{} in $\O{T}$ rounds of \CC{} deterministically.     
\end{proof}

While these simulation methods show that the round complexity can essentially be preserved when simulating \semiMPC{} and congested cliques algorithms on each other, the number of machines that \semiMPC{} needs in Theorem~\ref{thm:cc-mpc} is $n$, when the optimal number of machines in \semiMPC{} would be $m/n$ when $m \ll n^2$. We show, that by simulating \CONGEST{} \cite{peleg2000distributed} algorithms, the same round complexity can be achieved while using much fewer number of machines. The \CONGEST{} model is similar to the congested clique, however two nodes can communicate with each other, only if there is an edge between them in the actual graph. This limits the communication of each node and we use it to simulate multiple nodes in one machine.

\begin{theorem}\label{thm:congest-mpc}
	Let $\mathcal{A}_C$ be an algorithm that for an input graph with $n$ vertices and $m$ edges runs in $T$ rounds of \CONGEST{} while using, on each node, a local memory that is linear in the total messages that it receives. One can simulate $\mathcal{A}_C$ in $T$ rounds of \semiMPC{} using at most $\O{T \cdot m/n}$ machines.
\end{theorem}
\begin{proof}
	Note that the whole challenge in the proof of this theorem is for the case where $T \cdot m/n \ll n$, i.e., we do not have enough machines to simulate each node on one machine. Therefore, in the simulation, each \semiMPC{} machine needs to be in charge of more than one \CONGEST{} node. We give an algorithm for assigning \CONGEST{} nodes to \semiMPC{} machine such that it is possible for any machine to simulate the nodes assigned to it. Also, this assignment is known to all the machines.  Let $V_a$ denote the set of nodes assigned to machine $a$. Also, let $d_v$ denote the degree of node  $v$ in graph $G$. If $\Sigma_{v\in V_a} d_v$ is $\O{n/T}$ for any machine $a$,  it is possible for $a$ to simulate nodes in $V_a$. The reason is that in each round, communication of nodes in $V_a$ with other nodes is at most $\Sigma_{v\in V_a} d_v$ which does not violate the communication limit of \semiMPC{}. Also, during the algorithm the nodes receive at most  $T \cdot \Sigma_{v\in V_a} d_v$ which is $\O{n}$ and does not violate the memory limit of machines.
	
	To give the node assignment one first need to find the degree of vertices in $V$. For simplicity we use $m/n$ number of machines and we assume that edges are evenly distributed over these machines. (If it is not the case, we can simply achieve this in one round.) Let $d_{v, a}$ denote the degree of vertex $v$ in machine $a$ which is the number of edges in $a$ that one of their end point is $v$. Assume that vertices are labeled from $1$ to $n$, and let $l_v$ denote the label of vertex $v$. Also, machines are labeled from $1$ to $ m/n $ and $l_a$ denotes the label of machine $a$. In the first round for any vertex $v$, after computing $d_{v, a}$ in any machine $a$, send them to the machine  with label $\lfloor l_v/(m/n) \rfloor$. In this round the communication of each machine is $\O{n}$, since it receives at most $m/n$ messages for $\O{n/(m/n)}$ vertices. After this round of communication we are able to find the degree of all the vertices by summing up their partial degrees in one machine.
	
	Knowing the degree of vertices helps us to give the node assignment. Let $M= \O{T \cdot m/n}$ denote the number of machines. Also we assume that machines are labeled from $1$ to $M$. We sort vertices by their degrees in one machine and partition them to $M$ sets where the overall degree of vertices in each set is $\O{n/T}$. To achieve this for any $i\in n$ we assign the $i$-th node in the sorted list to the machine with label $i- M\cdot \lfloor i/M \rfloor$. Since the overall degree of vertices is $\O{m}$, the degree of vertices in each  machine is  $\O{m/M}$ which is  $\O{n/T}$. We also inform all the machines about this assignment. After this round each machine simulates the nodes assigned to it. Therefore, one can simulate $\mathcal{A}_C$ in at most $T+3$ rounds of \semiMPC{} using at most $\O{T \cdot m/n}$ machines.
\end{proof}

\section{Implications}
The simulation methods mentioned above, while not hard to prove, have significant consequences. For instance, by combining the recent algorithm of Jurdzi\'{n}ski and Nowicki \cite{jurdzinski2018mst} and Theorem~\ref{thm:cc-mpc}, we obtain the following which improves prior connectivity and MST algorithms on \semiMPC{} and \MPC{}. \cite{lattanzi2011filtering, bateni2017affinity}

\begin{corollary}
	There exists an $\O{1}$ round algorithm for computing the minimum spanning tree and graph connectivity on \semiMPC{}.
\end{corollary}

Furthermore, by combining the recent algorithms of Czumaj \etal{} \cite{czumaj2017round} and Assadi \etal{}~\cite{assadi2017coresets} with Theorem~\ref{thm:mpc-cc}, we get the first sub-logarithmic round algorithms for approximate maximum matching and vertex cover.

\begin{corollary}
	There exists an $\O{\log \log n}$ round algorithm for computing a $(1+\epsilon)$-approximate maximum matching on the congested clique.
\end{corollary}

\begin{corollary}\label{cor:vc}
	There exists an $\O{\log \log n}$ round algorithm for computing an $O(1)$-approximate vertex cover on the congested clique.
\end{corollary}



\section{Conclusion}
We formalized and compared the computational power of \semiMPC{} with other well-known models. Specifically, we showed that \semiMPC{} is almost equivalent to the congested clique model while incorporating another dimension to optimize for: the number of machines. 

\section{Acknowledgements}
We thank Sepehr Assadi for his helpful comments on an earlier version of this paper.

\bibliographystyle{alpha}
\bibliography{references}

\end{document}